\documentclass[review]{elsarticle}

\usepackage{hyperref}
\usepackage{listings}
\usepackage{amsmath}
\usepackage{amsthm}
\newtheorem{theorem}{Theorem}
\newtheorem{corollary}{Corollary}[theorem]

\usepackage{hyperref}
\newcommand{\myeq}[1]{\hfill{\refstepcounter{equation}(\theequation)\label{#1}}}

\usepackage{tkz-graph}
\usepackage{algorithm}
\usepackage[noend]{algpseudocode}
\usepackage{tikz}
\usetikzlibrary{arrows,automata}

\usepackage{multicol}
\journal{Journal of \LaTeX\ Templates}

\begin{document}

\begin{frontmatter}

\title{Graph based Data Dependence Identifier for Parallelization of Programs \tnoteref{mytitlenote}}

\author{Kavya Alluru}
\address{School of Computer Science and Engineering, Vellore Institute of Technology, Chennai}
\author{Jeganathan.L}
\address{Center for Advanced Data Science, Vellore Institute of Technology, Chennai}




\begin{abstract}
Automatic parallelization improves the performance of serial program by automatically converting to parallel program. Automatic parallelization typically works in three phases: check for data dependencies in the input program, perform transformations, and generate the parallel code for target machine. Though automatic parallelization is beneficial, it is not done as a part of compiling process because of the time complexity of the data dependence tests and transformation techniques. Data dependencies arise because of data access from memory required for the execution of instructions of the program. In a program, memory is allocated for variables like scalars, arrays and pointers. As of now, different techniques are used to identify data dependencies in scalars, arrays and pointers in a program. In this paper, we propose a graph based Data Dependence Identifier (DDI), which is capable of identifying all types of data dependencies that arise in all types of variables, in polynomial time. In our proposed DDI model, for identifying data dependence in a program, we represent a program as graph. Though many graphical representation of program exist, our approach of representing a program as graph takes a different approach. Also using our DDI model, one can perform basic transformations like dead code elimination, constant propagation, and induction variable detection. 
\end{abstract}


\end{frontmatter}

%\linenumbers

\section{Introduction}
\par Multicore processor, a chip with two or more processors, has completely replaced single core processor in personal computers. This development has put forth a challenge in the effective utilization of the processing power of multicore systems. Advancement in hardware technology always throws a challenge to the software community in the better utilization of the former. 
\par With multiple processes running on a personal computer, each process is made to run on each core which enhances the throughput of the system. But still, a single program can not utilize the processing power of multicore systems until unless it is converted to a parallel program. Converting a serial program to parallel involves breaking down the instructions of a serial program to multiple groups such that these groups can be executed in parallel.
\par Serial to parallel program conversion can be accomplished in two ways: manual or automatic. Manual conversion demands the programmer to have an insight in parallel architectures and parallel programming models, a quiet difficult task for programmer. Automatic conversion, known as automatic parallelization, is a much desired choice.  
\par A compiler which performs automatic parallelization is typically called as Parallelizing Compiler. A Parallelizing Compiler generally works in three phases: (1) Performs data dependence analysis on the source program. (2) Do transformations to remove dependencies in order to identify potential parallelism. (3) Generate parallel code appropriate to the target machine. Data dependence analysis plays a crucial role in deciding whether a program can be parallelized or not. Tests to identify data dependence in a program are discussed in Section 2.
\par Though multicore computers have become a household entity, automatic parallelization has not achieved such pinnacle. Automatic parallelization is not a part of traditional compilers. One reason is due to the high time-complexity of data dependence tests. Two statements in a program are said to be data dependent, if both the statements access same memory location i.e., a value computed in a statement is used by other statement and vice versa. Typically, in a program memory is allocated only using scalar variables, arrays, and pointers. In practice, there are many tests to identify different types of data dependencies(such as scalars, arrays and pointers) based on different approaches. As such there is no unique test to identify all kinds of data dependencies in a program. In this paper, we propose a unique model called Graph based Data Dependence Identifier (DDI) that performs data dependence testing of scalars, arrays and pointers. 

\section{Related Works}
  
Two instructions in a program are said to be \textbf{\textit{data dependent}} if both the instructions access a common memory location one atleast for write operation. In general, programs use scalar variables, arrays, and pointers to allocate memory. Data dependence analysis has been studied in all the three ways of memory allocation independently and extensively.
\subsection{Scalar Dependencies} A compiler other than translating a high level language program to machine understandable language, was further extended to optimize the program. In order to optimize a program, program has been represented as directed graph, to understand the data relationships between the statements in a given program. This concept was introduced by Allen \cite{allen}. To represent a program as graph, program statements are considered as nodes and control flow between these statements as edges. Many such graphs to analyze the data flow in a program were introduced by Kennedy \cite{kennedy}, Ullman \cite{ullman}.

\par This work was further extended to identify and reduce data dependencies in a program. By eliminating data dependencies, program statements can be executed in parallel. The initial work on data dependence elimination was restricted to scalar variables in the program. Kuck \cite{kuck} and Ferrante \cite{ferrante} introduced the concept of data dependence reduction by representing a program as dependence graph.  
\par To represent a program as graph Kuck \cite{kuck} considered, program components (assignment statements, for loop header, while loop header) as nodes. If data computed in a program component $c1$ is accessed by the other component $c2$, then an edge is added between the nodes $c1$ and $c2$, which shows existence of data dependence between $c1$ and $c2$. Five such dependencies are possible: loop, flow, output, input, and anti dependencies. The main objective of Kuck’s work is to reduce the dependencies between the scalar variables in order to identify potential parallelism there by reducing program's execution time.   
\par Ferrante \cite{ferrante} extended Kuck's model by representing a program as pair of graphs (data flow graph and control flow graph) together called as Program Dependence Graph(PDG). Here program statements are considered as nodes, data flow between the program statements are considered as edges in Data flow graph, control flow between the program statements are considered as edges in Control flow graph. Many program transformation techniques like code motion, loop fusion are performed using PDG. 
\par Though Kuck and Ferrante models help to reduce data dependencies in scalars, dependence that exist between the iterations of the loop (loop carried dependence or array dependences) and dependence due to pointers are not dealt here. Data flow analysis techniques used to identify and reduce data dependencies in scalars are discussed in \cite{aho}.
\subsection{Array Dependencies}The problem of identifying data dependence in arrays is formulated using linear equations and inequalities derived from loop subscripts and loop bounds. Many such tests have been proposed, the first one being the GCD Test\cite{uptal}. GCD Test considers the linear equations derived from loop subscripts as linear Diophantine equations and solves them using extended Euclid's algorithm. If an integer solution exists then the test proves the existence of data dependence. Banerjee's test is an extension of GCD Test that adds loop bounds as linear inequalities and solves the linear equations using Intermediate value theorem to check if any real solutions exists. Both GCD test and Banerjee test are approximate methods as they converge to 'dependence exists' in case if the test is inconclusive. I-Test is an extension of Banerjee's tests. It is based on the observation that real solutions predicted by banerjee's test are integer solutions\cite{kong}. All these tests, in case of multidimensional arrays, solves each subscript expression independently to check whether an integer solution exists or not. Omega test\cite{omega} uses Fourier Motzkin Variable Elimination(FMVE) method to solve the system of linear equations and inequalities for an integer solution.This test is more accurate but has the worst case exponential time complexity. Range test\cite{range} checks if a data dependence exists if the loop subscripts are non-linear. NLVI Test \cite{nvli} solves the non-linear and symbolic expressions derived from loop subscripts. This test experimentally proves that it can handle complex loop bounds. QP Test\cite{qptest} proves that the non-linear loop subscripts that are in quadratic form can be solved using quadratic programming. Many other dependence tests exist in the literature\cite{exact}, \cite{practical}, \cite{power}.
\subsection{Pointer Dependencies} Solving data dependence in pointers is considered as separate research area. Alias analysis\cite{alias} and shape analysis\cite{shape} are two popular techniques used to parallelize loops with pointers. \cite{skeleton} transforms the code with pointers to a simpler form.
\subsection{Automatic Parallelization Tools}
SUIF\cite{suif} and Polaris\cite{polaris} were the earliest tools designed for Symmetric Multiprocessor Systems, convert serial FORTRAN program to parallel form. Cetus\cite{cetus} a successor of Polaris, coverts C programs to parallel and is designed for multicore systems. Intel compiler \cite{intel} automatically identifies the loops that can be parallelized and partitions the data accordingly. Any automatic parallelization tool has to perform data dependence testing to convert serial program to parallel program. All the tools mentioned here, make use of data flow analysis for identifying dependencies in scalars, linear equation formulation for array dependencies, alias analysis techniques for pointer dependence. As of now, all the existing tools, make use of different techniques to identify data dependence in scalars, arrays and pointers. As such there is no single model with which one can identify all types of data dependencies.
\par  The main objective of the paper is to design a graph based model called Data Dependence Identifier (DDI) with which one can identify any type of data dependence in polynomial time. In contrast to Kuck and Ferrante model of representing a program as graph where \lq{program statements}\rq  are considered as nodes, our model takes a different approach, we consider variables in the program as nodes and the edges between these variables are drawn based on the mode of accessing the variables from memory. With this novel approach of representing a program as graph, we could identify all types of data dependencies in polynomial time. In addition our model can also perform compiler optimizations such as dead code elimination, constant propagation, and induction variable detection. 
\section{Parameterization of program}
\par As mentioned earlier, our model Data Dependence Identifier(DDI) can identify all types of data dependencies that exist in a program. Since the core concept behind our DDI model, is the graphical representation of a program with different approach. For that purpose we view a program as a structure with finite number of components called as parameters of the program. In this section, we discuss in detail about the parameterization of the program.
\subsection{Program instructions}
\par A program is a sequence of step-by-step instructions when executed gives the desired output. It is a combination of instructions and data. Data is assigned with the help of variables initialized in the program. An instruction in a program, irrespective of the programming language, will fall into any one of the following types:\\
\textbf{\textit{Assignment Instructions}}: In an assignment instruction, value is assigned in two different ways: a constant can be assigned to a variable, for example $a=5$. A value stored in a variable can be assigned to another variable, for example $a=b$.\\
\textbf{\textit{Arithmetic Instructions}}: In an arithmetic instruction, input is read from one or more variables or from a programmer and the output is assigned to a variable after performing the required arithmetic operations. Example: $a=b+c$, $a=b+3$, $a=b+c+a$, $a=c-b*5$.\\
\textbf{\textit{Conditional instructions}}: Conditional instructions reads data from one or more variables or constant values, performs logical operations and decides either True or False. In conditional instructions, output is not necessarily assigned to any variable .\\ 
\textbf{\textit{Iterative Instructions}}: Iterative instructions will execute repeatedly until the specified condition fails .  \\ 
\textbf{\textit{Control transfer Instructions}}: These instructions switches the control from one instruction to another. For example, instructions like break, continue, goto, jump etc are control transfer instructions.  \\
\textbf{\textit{Input Instructions}}: Input instructions reads the data from input devices and writes to memory. In example: $ \lq Read \; a,b \rq$, data is read from programmer through the input device and stored in memory location $a$,$b$ respectively.\\
\textbf{\textit{Output Instructions}}: Output instructions reads the data from memory and writes to an output device. Example: $print(a)$. \\
\textbf{\textit{Function calls}}: Function calls are a request made to another routine that executes a predetermined task.
%Following observations were made based on the instructions behavior on memory.
%\begin{itemize}
%\item Arithmetic instructions read data from the memory and writes output to memory.
%\item Conditional instructions read the data from the memory, no write is made to memory.
%\item Input instruction reads data from input device and writes to memory. Output instruction reads data %from memory and writes to an output device.
%\item In assignment instructions, data is written to a memory location.
%\item Constant values initialized in an instruction exists in all types of instructions.
%\end{itemize}
\subsection{Categorization of Instructions based on Memory}
\par Knowing all the types of instructions with which a computer program is made, one could observe that, some instructions access the memory to read data from a variable, some instructions access the memory to write data to a variable. Some instructions access the memory for both reading a data and for writing the data. Some instructions do not access the memory at all. Based on the way an instruction accesses the memory , we categorize an instruction as Memory Access Instruction (MAI) and Non Memory Access Instruction (NMAI).\\
\textbf{Memory Access Instruction (MAI)}: Instructions that access the memory to perform the required operation. For example: arithmetic, conditional, input, output instructions are MAI. \\
\textbf{Non Memory Access Instruction (NMAI)}: Instructions that do not access the memory at all. For example, control transfer instructions like break, jump fall under this category. \\
In MAI, some instructions access the memory for read operation alone, some for write operation alone, some for both read and write. Based on the operations, performed by MAI on memory, these instructions are further classified into three categories : MA-READ, MA-WRITE, MA-READWRITE. \\
\textbf{MA-READWRITE(MARW)}: Instructions that access the memory for both read as well as write operations come under this category. For example, in Arithmetic instruction: $`c=a+b'$, data is read from memory locations $a$ and $b$ and written to a memory location $c$. In assignment instruction: $`a=b'$, data is read from memory location $b$ and written to $a$. \\
\textbf{MA-READ(MAR)}: Instructions that perform only read operation but no write operation are classified as MA-READ. For example, in conditional instruction: $`if(a>b)'$ data is only read from memory locations $a$ and $b$ but the output is not written to any variable. In the output instruction: $print(a)$, data is read from memory location $a$ and sent to an output device. \\
In conditional instruction: $if(a>b)$, data is read from memory locations $a,b$ and sent to processor for further computation. In instruction $print(a)$, data is read from memory and send to an output device. In general, data is read from memory and send to other Hardware Units(HU) in the computer system like processor or output devices.  \\
\textbf{MA-WRITE(MAW)}: Instructions that perform only write operation but no read operation are classified as MA-WRITE. For example, in assignment instruction $`a=5'$, a constant value is written to a memory location $a$. The instruction $a=5$ conveys that the programmer is writing 5 into the memory location $a$. $a=5$ differs from the instruction $a=b$ where data is read from location $b$ and written to location $a$. In that sense, we deem that the instruction $a=5$ means that the constant 5 is read from the programmer(PR) and written to the location $a$. In the input instruction $`read \; a,b'$ data is read from input device and written to memory locations $a$ and $b$. 
The following diagram summarizes the classification of instructions based on the memory access behavior.
\includegraphics[scale=0.4]{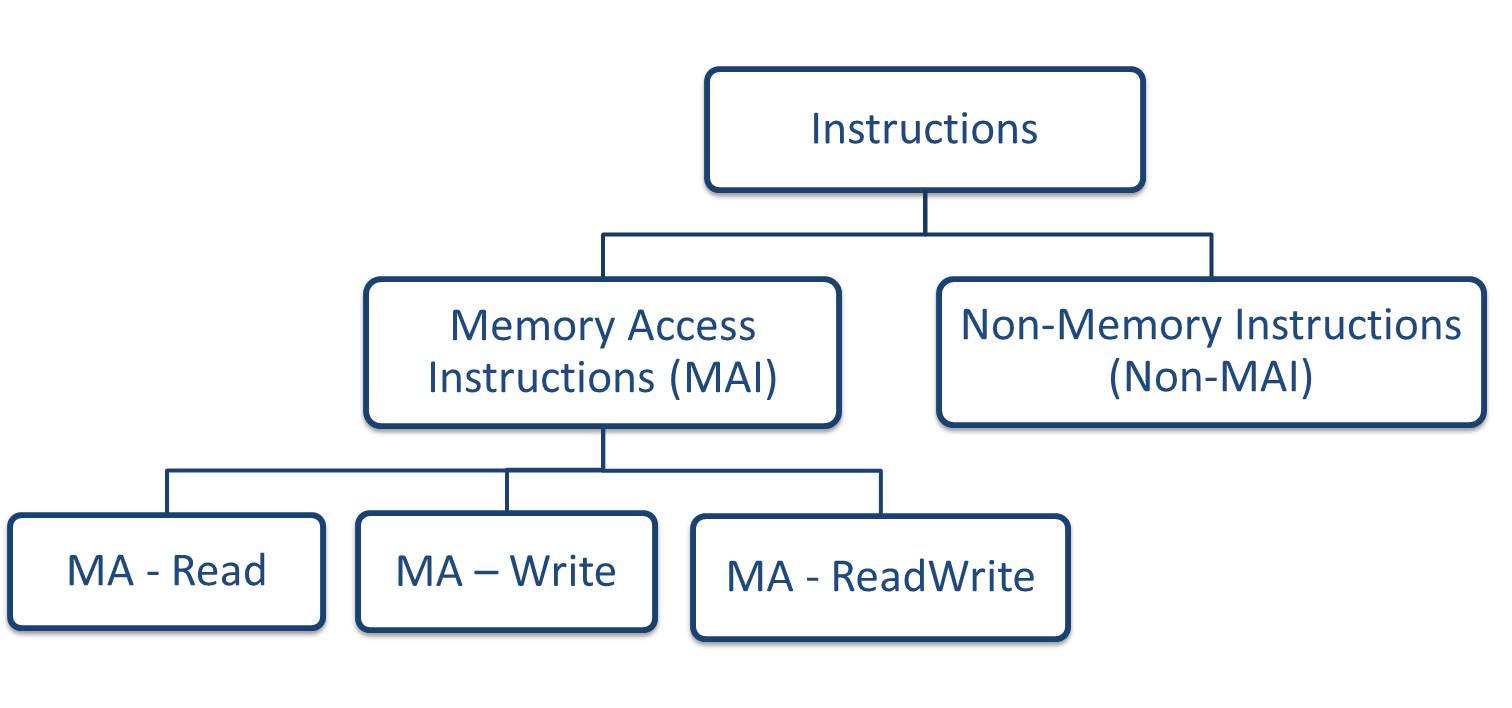}
\subsection{Representation of MAI as a pair of variables}
As mentioned, typically every MARW instruction will read the data from more than one memory location and write the data into a memory location. This gave an idea that one can represent a MARW instruction $i$ as an ordered pair of set of variables $(R,W)$ where R is a set that contains all the variables from which the instruction $i$ reads the data and W is a set with a single variable to which $i$ writes the data. Eventually, this representation of MAI as ordered pair of set of variables can be applied to MAR instructions and MAW instructions also. A MAR instruction $i$ can be thought of $[R,HU]$ where data is read from variables in R and written to a hardware unit $HU$.  A MAW instruction $i$ can be thought of $[HU,W]$ or $[PR,W]$ where data is read from either programmer(PR) or Hardware Unit(HU) and written to a variable in W. For example,
\begin{itemize}
\item A MA-READWRITE instruction $`c=a+b'$ is written as pair $[\{ a,b \},\{ c \} ]$. Data is read from variables $a$ , $b$ and output is written to $c$.
\item MA-READ instruction $`if(a >b)'$ is represented as $[\{ a,b \},\{ HU \} ]$. Data is read from $a$,$b$ and send to a Hardware Unit(HU).
\item MA-WRITE instruction $`a=5'$ is represented as $[\{ PR \},\{ a \} ]$. A value initialized in the instruction by the programmer(PR) is read and written to variable $a$.
\end{itemize}
Thus we conclude that every MA instruction can be represented as pair of set of variables $(R,W)$.
\subsection{Parameterization of program}
%\par The main objective of this paper is to propose a Graph based Data Dependence Analyzer (GDA) model %that can identify data dependencies in scalars, arrays and pointers of a given program. Given a serial %program $P$, is converted to a unique graph called $G_P$. $G_P$ is tested for data dependencies with the %help .  \\
Thus, every program P is a collection of finite set of instructions $I=\{ i_1, i_2,...i_n \} $, which operate on the data stored in finite set of memory allocations $V= \{ v_1, v_2,...v_m \} $. $MAI$ is a finite set which contains all MA Instructions where $MAI \subseteq I$.
Based on the discussion so far, any program P can be parameterized with $I, V ,W ,HU ,PR$. Accordingly we write $P$ as $P(I,V \cup \{ HU,PR \} ,MAI)$. We consider $HU$ and $PR$ as fixed variables as data is read from and written to $HU$ and data is read from $PR$ alone. Thus, we conclude that a program is made up of all these following components:
\begin{itemize}
\item Set $I$, finite set of instructions $\{i_1, i_2,...i_n\} $ 
\item Set $V$, finite set of memory allocations or variables $\{v_1, v_2,...v_p\}$ 
\item Set $MAI$, finite set of MA instructions where $MAI \subseteq I$. An instruction $i \in MAI$ can be written as ordered pair $[R,W]$ where $R,W \subseteq V$.   
\item $HU$ represents the set of hardware units i.e. input devices, output devices, processor and any other hardware unit in the computer system.
\item $PR$ is the set of constant values initialized in the program P by the programmer.
\end{itemize}
%\par Every MA-READWRITE instruction say $I_s$ is a mapping from $V_1, V_2,...,V_n$ to $V_m$. Information is read from one or more locations say $V_1, V_2,...,V_n$ and written to another location say $V_m$. That is, $I_s:V_1, V_2,...,V_n \rightarrow V_m$ such that $I_s(a_1,a_2,...a_n)=b$, which conveys that one or more locations are read and written to another location $b$. 
%\par In MA-READ instruction $I_s$ data is read from one or more locations say $V_1, V_2,...,V_n$ and send to another hardware unit like processor or input output device in the computer system $HU$. 
%\par In MA-WRITE instruction $I_s$ data is written to a single memory location say $V_m$ from $HU$ or a constant value say $CO$ initialized in the instruction $I_s$ of the program. 
\section{Directed graph Representation of a program}
\par As mentioned earlier, main aim of the paper is to automate the identification of data dependencies with which one can decide whether a sequential program can be parallelized or not. For that purpose, we first transform a program P to an equivalent directed graph called graph of P written as $G_P$.\\
\textbf{Labeled Directed Graph}: A labeled directed graph G=(N,E,L) where $N$ is the finite set of nodes $\{n_1,n_2,...n_k \} $. $E=\{e_{ij}= (n_i,n_j)| n_i,n_j \in N\}$, where the pair $(n_i,n_j)$ indicates a directed edge that starts from node $n_i$ and terminates at $n_j$. A mapping $L:E\rightarrow S$ i.e. every edge is assigned a label from the elements of S, called the set of labels. \\
All instructions in a given program P are indexed sequentially with the positive integers $1,2,...n$. First instruction in a program is indexed as 1, second instruction as 2 and so on. For $i_n \in I$, we call index of $i_n=n$. We know that every instruction $i_n$ can be written as the pair $[R,W]$. In other words, $index[i_n]=index([R,W])=n$.\\
A program $P=(I,V\cup\{PR,HU\},MAI)$ is transformed into a directed labeled graph $G_P=(V\cup\{PR,HU\},E,L)$ as follows:
\begin{itemize}
\item Set of nodes of $G_P$ are the set of variables $V \cup \{PR,HU \}$.
\item For every ordered pair of sets $(R,W) \in MAI$, we include the edges $\{[r,w])| \forall r\in R, w\in W\}$. 
\item Every edge in $G_P$ is labeled with elements from label set S which contains indices of instructions in I. $L:E\rightarrow \{1,2,...n\} $ such that $L((r,w))=k$ if $index([R,W])=k$ such that $(r,w) \in E, r \in R, w \in W$.
\end{itemize}
We use the notation $(.,.)$ to represent the edges of the graph and $[.,.]$ indicates the pair of sets R,W for representing memory access instructions.\\ 
 In example 1, $I= \{i_1,i_2,i_3,i_4 \}$, $V=\{ a,b,c,d \}$ and MAI=I as all instructions in program P are memory access instructions. To construct $G_P$, V acts as nodes N. For instruction $1:[ \{ a,b \} ,\{ c \} ]$, we include the edges $(a,c)$ and $(b,c)$ with labels $L((a,c))=1$ and $L((b,c))=1$ are added to $G_P$. For instruction $2:[ \{ a,PR \} ,\{ d \} ]$, edges $(a,d)$ and $(PR,d)$ with labels $L((a,d))=2$ and $L((PR,d))=2$ are added. For instruction $3:[ \{ c,d \} ,\{ HU \} ]$, edges with label $L((a,HU))=3$ and $L((d,HU))=3$ are added. For instruction $4:[ \{ a,PR \} ,\{ b \} ]$, edges with label $L((a,b))=4$ and $L((PR,b))=4$ are added to $G_P$.

\begin{tabular} { l c r }
 \hline
 \multicolumn{3}{c}{Example 1} \\
 \hline            
 (a) & (b) & (c)\\
 \hline
\begin{lstlisting}[mathescape=true]
void func1(int a,int b)
{
$1:$ c=a+b;
$2:$ d=a-10;
$3:$ if(c>d)
$4:$ b=a+10
}
\end{lstlisting}   & 

\begin{tikzpicture}[scale=2]
\GraphInit[vstyle=Normal]
%\tikzset{ LabelStyle/.style = { rectangle, rounded corners, draw}}
\SetUpEdge[style={->}]
  \SetGraphUnit{3}
  \Vertex{b}
  \Vertex[x=-1.4,y=0.0]{a}
  \Vertex[x=-1.4,y=-1.0]{c}
    \Vertex[x=0,y=-1.2]{d}
  \Vertex[x=-0.6,y=-2.0]{HU}
  \Vertex[x=0.8,y=-0.6]{PR}
 % \WE(B){A}
 % \EA(B){C}
  \Edge[label = $1$](a)(c)
  \Edge[label = $1$](b)(c)
   \Edge[label = $2$](PR)(d)
   \Edge[label = $3$](c)(HU)
   \Edge[label = $3$](d)(HU)
       \Edge[label = $4$](PR)(b)

      \tikzset{EdgeStyle/.style = { ->,bend right}}
  \Edge[label = $2$](a)(d)
      \tikzset{EdgeStyle/.style = { ->,bend left}}

  \Edge[label = $4$](a)(b)

%  \Loop[dist = 4cm, dir = NO, label = 5](A.west)
 % \Loop[dist = 4cm, dir = SO, label = 6](C.east)
  \end{tikzpicture}   &   
\scalebox{0.8}{
\begin{tabular}{| c | c | c | c | c | c | c |}
\hline
& a & b & c  & d& PR & HU  \\  \hline
a & &$4$ &$1$  &  $2$& &   \\   \hline
b & & & $1$ &  &  &   \\   \hline
c & & &  & & & $3$  \\   \hline
d & & &  & & &$3$    \\  \hline   
PR & &$4$ &  &$2$ & &     \\  \hline   
HU & & &  & & &    \\  \hline   
\hline
\end{tabular}}
\\  \hline
\end{tabular}
The adjacency matrix of graph $G_P$ is shown in example 1(c), rows gives the \textbf{\textit{read}} information about the variables and columns gives the \textbf{\textit{write }} information. Scanning column $c$ of the matrix tells that variable $c$ is accessed for $`write'$ in instruction $1$ and row of $c$ shows variable $c$ is accessed for $`Read'$ in instruction $3$. \\
 The procedure by which we convert $P(I,V\cup \{PR,HU \} ,W)$ into a simple edge labeled graph $G_p(N\cup \{PR,HU \},E,L)$ is discussed in algorithm 1.  Graph $G_p$ is represented using adjacency matrix.
\begin{algorithm}
\caption{Convert Program $P$ to Directed edge-labeled graph $G_P$}
\begin{algorithmic}[1]
\Procedure{Program to Graph}{}\newline
\textbf{Input:} Program $P(I,V\cup \{PR,HU \})$ \newline
\textbf{Output:} Graph $G_p(N\cup \{PR,HU \},E,L)$  \label{euclid}
\For{each instruction $[R,W] \in I,index[R,W]=k$}
 \If{MAI-verification()}
      \For{every $ r \in R$ and $w \in W$}
          \State{$E=E \cup \{(r,w) \}$}
          \State{$L([r,w])=k$}
      \EndFor
  \EndIf    
\EndFor

\EndProcedure
\end{algorithmic}
\end{algorithm}
The procedure works as follows. Loop in line 2 takes each instruction in a chronological order and line 3 calls a MAI-verification procedure that checks whether an instruction is a Memory Access Instruction (MAI) or not. If instruction $i_k:[R,W]$ is MAI, then for each $r \in R$ an edge $(r,w)$ with label $k$ is added to graph $G_p$, shown in lines 4-6. \\
The running time of this algorithm depends on two components: procedure to check whether an instruction is MAI or not and number of edges added to $G_p$. MAI-verification procedure is called for every instruction of I. Every programming language, will have few keywords, say jump, break, exit etc, such that the presence of these keywords in an instruction does not require the access of memory. Let $`m'$ be the number of keywords in a programming language which does not require any memory access. $`m'$ is a constant for a specific programming language. If each MAI-verification procedure call has to go through maximum of \textit{m} comparisons to conclude the category of an instruction and there are $n$ instructions in a program, procedure takes $O(n)$. For every instruction, set of edges are added to $G_p$ so totally for all instructions $E$ edges are appended to $G_P$, adding each edge takes a constant time, $E$ edges take $O(E)$.
\begin{theorem}
 Given a program P, there exists an unique simple edge labeled graph $G_P$ that corresponds to $P$. 
 \end{theorem}
\begin{proof}
 Suppose there corresponds two graphs $G_P^\prime $ and $G_P^{\prime\prime}$ to the program P. \\
Let $G_P^\prime = (V^\prime \cup \{PR,HU \},E^\prime,L^\prime) $  \\
 $G_P^{\prime\prime} = (V^{\prime\prime} \cup \{PR,HU \},E^{\prime\prime},L^{\prime\prime}) $ \\
First we show that, If $a \in V^\prime $ then $a \in V^{\prime\prime}$ and vice versa. \\
Let $a \in V^\prime$, \\
$ \iff a $ is a vertex in $G_P^\prime$ that corresponds to program P. \\
$ \iff a $ is a variable in program P. \\
$ \iff a $ is a vertex in $G_P^{\prime\prime}$ since $G_P^{\prime\prime}$ is a graph that corresponds to program P. \\
$ \implies a  \in V^{\prime\prime}$ \\
Therefore, all the elements of $V^\prime$ are also in $V^{\prime\prime}$. \\
$ \implies $ All the elements of $V^\prime \cup \{ PR,HU \} $ are also in $V^{\prime\prime} \cup \{ PR,HU \} $. \\
 Hence, number of elements in $V^\prime \cup \{ PR,HU \} $ is same as number of elements in $V^{\prime\prime} \cup \{ PR,HU \} $.     \myeq{1}  \\
 Now, we show that $G_P^\prime$ is isomorphic to $G_P^{\prime\prime}$. 
 \begin{enumerate}
 \item Consider the identity function $f: V^\prime \cup \{ PR,HU \} \rightarrow V^{\prime\prime} \cup \{ PR,HU \} $ such that $f(a)=a$,where $a \in V^\prime$, $f(PR)=PR$ and $f(HU)=HU $. \\
 By equation 1, $f$ is bijective since $f$ is identity function.
 \item If $(a^\prime ,b^\prime) \in V^\prime $ then $(f(a^\prime) ,f(b^\prime)) \in V^{\prime\prime} $ \\
 Let $(a^\prime ,b^\prime) \in V^\prime $, there is an instruction in P such that data is read from $b^\prime$ and written in the location $a^\prime$. \\
$ \implies $ There is an edge $(a^\prime ,b^\prime) \in V^{\prime\prime} $ since $G_P^{\prime\prime} $ is the graph that corresponds to P. \\
 $ \implies$ the edge $(f(a^\prime) ,f(b^\prime)) \in V^{\prime\prime} $. \\
 Therefore,  $G_P^\prime \cong G_P^{\prime\prime} $    
   \end{enumerate}
  \end{proof}
\section{Data Dependence Identifier}
Two instructions $i_j$ and $i_k$ in a program are said to be \textit{data dependent} if both access the same memory location one for $`read'$ and the other for $`write'$. There are four types of data dependencies \cite{dd1},\cite{dd2}:
\begin{itemize}
\item \textbf{Flow dependence} exists between two instructions $i_j$ and $i_k$ if $i_j$ writes to a memory location that $i_k$ reads later.For example, in the code given below:
 \begin{lstlisting}[mathescape=true]
$I_1:$ c=a+b;
$I_2:$ d=c-b;
\end{lstlisting}  
Instruction $I_2$ uses the value computed by $I_1$. This is called as flow dependence as the data flows from $I_1$ to $I_2$.
\item \textbf{Anti dependence} exists between two instructions $i_j$ and $i_k$ if $i_j$ reads from a memory location that $i_k$ writes later. In the code given below, 
\begin{lstlisting}[mathescape=true]
$I_1:$ c=a+b;
$I_2:$ a=b+3;
\end{lstlisting}  
Instruction $I_2$ computes $a$, $I_1$ has read the old value of $a$ before $I_2$. This is called as anti dependence.
\item \textbf{Output dependence} exists between two instructions $i_j$ and $i_k$ if $i_j$ writes to a memory location that $i_k$ also writes later. 
\item \textbf{Input dependence} exists between two instructions $i_j$ and $i_k$ if $i_j$ reads from a memory location that $i_k$ also reads later. 
\end{itemize}  
Data dependence testing is a process to identify all such dependencies in a given program. In this section, we discuss the process by which the graph $G_P$ identifies the data dependencies in a given program P. For this reason, we call our $G_P$ as DDI. Algorithm 2 illustrates the process how DDI identifies whether data dependence among the instructions exists in a program. The input to the algorithm is graph $G_p$ of program P as discussed in previous section. The algorithm outputs the type of data dependence that exists. 
\begin{algorithm}
\caption{Identification of Data Dependencies in a program}
\begin{algorithmic}[1]
\Procedure{Data Dependencies}{}\newline   \label{euclid}
\textbf{Input:} Graph $G_p(N,E,L)$
 \For{every $v \in N.G_p$}
   \If{$(L((u,v)) \neq NULL) \&\& (L((v,u^1)) \neq NULL) $}
   		% \State{ where $u,u^1 \in N.G_p $}
   		\If{$(L((u,v)) < (L((v,u^1)) $}
          \State{Flow Dependence exists}
    \ElsIf{$(L((u,v)) > (L((v,u^1)) $}
          \State{Anti dependence exists}
    \EndIf
    \EndIf
       \If{$(L((u,v)) = k) \&\& (L((u^1,v)) = j) $}
       \If{$ k \neq j $}
               \State{Output dependence exists}
    \EndIf
    \EndIf
       \If{$(L((v,u)) = k) \&\& (L((v,u^1)) = j) $}
       \If{$ k \neq j $}
               \State{Input dependence exists}
    \EndIf
    \EndIf

\EndFor
\EndProcedure
\end{algorithmic}
\end{algorithm}
\par For every node $v \in N.G_p$(node set of $G_P$) the following condition is checked in lines 3-7: whether $v$ has an incoming edge $(u,v)$ and an outgoing edge $(v,u^\prime)$. The existence of these edges $((u,v),(v,u^\prime))$ confirm that data dependence exists between the respective instructions since variable $v$ is accessed for both $`Read'$ and $`Write'$. An incoming edge to $v$ with label $k$ interprets memory location $v$ is accessed for $`Write'$ in instruction number $k$ and outgoing edge with label $j$ from $v$ interprets memory location $v$ is accessed for $`Read'$ in instruction number $j$.\\
 Based on the labels of the edges: $L(u,v)$ and $L(v,u^\prime)$, we identify the nature of data dependence.
\begin{enumerate}
\item \textbf{Flow dependence:} If ${L((u,v)) < L((v,u^\prime))}$ means that $v$ is accessed for $`Write'$ first and then for $`Read'$.
\item \textbf{Anti dependence:} If ${L((u,v)) > L((v,u^\prime))}$ means that $v$ is accessed for $`Read'$first and then for $`Write'$.
\item \textbf{Output dependence:} If there exists edges $L((u,v))=j$ and $L((u^\prime,v))=k$ means that $v$ is accessed for $`Write'$ in instructions $j$ and $k$.
\item \textbf{Input dependence:} If there exists edges $L((v,u))=j$ and $L((v,u^\prime))=k$ means that $v$ is accessed for $`Read'$ in instructions $j$ and $k$.
\end{enumerate}
The running time of this algorithm is $O(|N|^2)$ as the entire row and column of each variable $v$ has to be scanned, where N is the number of number of variables used in the program. For the correctness of the above algorithm we prove the following theorem.
\begin{theorem}
Let P be a program and $G_P$ be the corresponding graph of P. If $G_P$ has a path joining the vertices $ v_i \xrightarrow{\text{l}} v_j \xrightarrow{\text{m}} v_k$, $v_j \notin \{ PR,HU \} $, then the pair of instructions $(i_l,i_m)$ has the data dependence problem. 
\end{theorem} 
\begin{proof}
A pair of instructions in a program will have the data dependence if both instructions access the same memory location for both the operations : read and write. \\
Let $G_P$ has a path joining $v_i$ and $v_k$, of length 2 such that $v_i \xrightarrow{\text{l}} v_j \xrightarrow{\text{m}} v_k$ i.e., there is an edge $(v_i,v_j) \in E $ with the label $l$ and an edge $(v_j,v_k) \in E $ with label $m$. 
\begin{enumerate}
\item $(v_i,v_j) \in E $ with label $l \implies $ $l^{th}$ instruction of P instructs that the memory location $v_i$ is accessed for `Read' and $v_j$ is accessed for `Write'.  
\item $(v_j,v_k) \in E $ with label $m \implies $ $m^{th}$ instruction of P instructs that the memory location $v_j$ is accessed for `Read' and $v_k$ is accessed for `Write'.  
\end{enumerate}
1 and 2 $ \implies $ the $l^{th}$ instruction accesses $v_j$ for `Write' and the $m^{th}$ instruction accesses $v_j$ for `Read'. \\
$\implies $ $l^{th}$ and $m^{th}$ instruction of P accesses the same memory location for both `Read' and `Write'.\\
$ \implies $ Instructions $(i_l,i_m)$ has the data dependence.
\end{proof}
\begin{corollary}
If $l < m$ then the instructions $(i_l,i_m)$ have flow dependence.
\end{corollary}
\begin{corollary}
If $l > m$ then the instructions $(i_l,i_m)$ have anti dependence.
\end{corollary}
\begin{corollary}
If $L((v_i,v_j))=l$ and $L((v_k,v_j))=m$ and  then the instructions $(i_l,i_m)$ have output dependence.
\end{corollary}
\begin{corollary}
If $L((v_j,v_i))=l$ and $L((v_j,v_k))=m$ and  then the instructions $(i_l,i_m)$ have input dependence.
\end{corollary}
Since the proof of corollary are immediate, we have not described the proof.\\
\textbf{Note 1:} In a path of length 2: $v_i \xrightarrow{\text{l}} v_j \xrightarrow{\text{m}} v_k$, and $v_j \in {HU}$, then the pair of instructions $(I_l,I_m)$ need not have  the data dependence since HU is the vertex used for Read/Write from the hardware devices. Two different instructions with one instruction reading from a device and another instruction writing into the device may not have any link at all. For example,instructions print a and Read c do not have any connection among them. \\
\textbf{Note 2:} If $G_P$ has a path of length k, with the sequence of labels of the edges of the path as: $l_1,l_2,l_3,...l_k$. Then all the possible pairs of instructions $I_{l_1},I_{l_2},I_{l_3},...I_{l_k}$ has the data dependence.
 
\subsection{Data Dependence in arrays}
\par One of the fundamental and crucial job of data dependence identifier is to decide whether an array can be parallelized or not i.e., whether the elements in an array can be split into groups. Generally, loops are used to traverse elements in an array. If each index position of the array is accessed by each instance of the loop, then there exists no data dependence in that array. A variable(index position of an array) accessed for `Write' in one instance of execution is `Read' in other instance of loop by other instructions or vice versa, data flows between iterations of loop. If flow dependence among different instances of execution of loop exist, such data dependence is called \textit{loop carried dependence} or \textit{array dependence}. A loop can not be parallelized if such dependence exists. \\
In this section, we discuss how our model DDI, solves the complicated problem of identifying data dependencies in arrays with the help of an example.  \\
 \textbf{Data Dependence in One Dimensional Arrays} \\
 With our DDI, we can easily identify the data dependence in loops. As usual, the indices of array are considered as individual variables.\\
 Consider the program in Example 2. Normally a sequential loop executes in an increasing order of index numbers following the order of instructions. Here there are three instructions $4$, $5$, $6$ and the loop indexing goes from 2 to 4. Consider three instances of execution:
 \begin{lstlisting} [mathescape]
$4.1:$ a[2]=b[2]+c[2];
$5.1:$ a[3]=a[1]+c[1];
$6.1:$ c[1]=b[2];

$4.2:$ a[3]=b[3]+c[3] ;
$5.2:$ a[4]=a[2]+c[2];
$6.2:$ c[2]=b[3];

$4.3:$ a[4]=b[4]+c[4];
$5.3:$ a[5]=a[3]+c[3];
$6.3:$ c[3]=b[4];
\end{lstlisting}
The statements within the loop are denoted as $i.k$, $1\leq i \leq n$,  where $i$ represents the instruction number and $k$ represents the instance of execution. \\
We convert the loop instructions of P to $G_P$. In example 2, $I= \{ 4,5,6 \}$, $V= \{ a[0],c[0],a[1],b[1],c[1],a[2],b[2],c[2],a[3],b[3],c[3],a[4],b[4],c[4],a[5] \}$, $MAI=I$ as all instructions are memory access instructions. To construct $G_P$, all variables in $V$ acts as nodes $N$. For instruction $4.1:[ \{ b[1],c[1] \}, \{a[1]\} ]$, edges with label $L(([b[1],a[1]]))=4.1$ and $L((c[1],a[1]))=4.1$ are added to $G_P$. For instruction $5.1:[ \{ a[0],c[0] \}, \{a[2]\} ]$, edges with label $L((a[0],a[2]))=5.1$ and $L((c[0],a[2]))=5.1$ are added to $G_P$. Similarly for all instances of execution, edges are added to $G_P$ as shown in example 2(b).\\
In the above given instances of execution, variable $a[1]$ computed in the instruction $4.1$ is read in other instance of another instruction $5.2$. This is shown in $G_P$ as edges $(b[2],a[2])$ and $(a[2],a[4])$, which is a flow dependence. Similarly, variable $a[3]$ computed in $5.1$ is read in $5.3$. This is shown in $G_P$ as edges $(c[1],a[3])$ and $(a[3],a[5])$.   \\

\begin{tabular}{ l c r  }
 \hline
 \multicolumn{2}{c}{Example 2} \\
 \hline
 (a) & (b) \\
 \hline
 \begin{lstlisting} [mathescape]
void add()  
{
for($\overbrace{i=2}^{I_2}$;$\overbrace{i<5}^{I_2}$;$\overbrace{i++}^{I_3}$)
$4:$ a[i]=b[i]+c[i];
$5:$ a[i+1]=a[i-1]+c[i-1];
$6:$ c[i-1]=b[i];
}
\end{lstlisting}   &\begin{tikzpicture}[scale=1.8]
\GraphInit[vstyle=Normal]
%\tikzset{ LabelStyle/.style = { rectangle, rounded corners, draw}}
\SetUpEdge[style={->}]
  \SetGraphUnit{3}
  \Vertex {a[2]}
  \Vertex[x=-1.2,y=-0.9]{b[2]}
    \Vertex[x=-0.4,y=-1.5]{c[2]}
        \Vertex[x=-0.9,y=-2.5]{b[3]}

  \Vertex[x=0.2,y=-3.0]{a[3]}
  \Vertex[x=1.2,y=-4.0]{a[1]}

  \Vertex[x=-0.6,y=-4.0]{c[1]}
        \Vertex[x=2.5,y=-3.0]{c[3]}
        \Vertex[x=1.8,y=-2.0]{a[4]}
        \Vertex[x=3.0,y=-1.0]{b[4]}
        \Vertex[x=2.0,y=-0.2]{c[4]}
          %      \Vertex[x=3.7,y=-0.2]{b[5]}
   
       \Vertex[x=3.5,y=-1.2]{a[5]}
% \tikzset{VertexStyle/.style = {shape=circle,fill=black}}
 %\Vertex [x=0,y=1.0]{s}
 %\tikzset{VertexStyle/.style = {shape=rectangle,fill=black}}
 %\Vertex [x=0.8,y=-1.0]{s1}

    \tikzset{EdgeStyle/.style = { ->,bend left}}
  \Edge[label = $4.1$](b[2])(a[2])
  \Edge[label = $4.1$](c[2])(a[2])
  \Edge[label = $5.1$](a[1])(a[3])
    \Edge[label = $5.1$](c[1])(a[3])
      \Edge[label = $4.2$](b[3])(a[3])
  \Edge[label = $5.3$](c[3])(a[5])
  \Edge[label = $4.2$](c[3])(a[3])
  \Edge[label = $6.2$](b[3])(c[2])
  \Edge[label = $6.3$](b[4])(c[3])
  %\Edge[label = $L_6.3$](b[5])(c[4])

   \tikzset{EdgeStyle/.style = { ->,bend right}}
  
    \Edge[label = $5.2$](a[2])(a[4])
        \Edge[label = $5.2$](c[2])(a[4])

    \Edge[label = $4.3$](b[4])(a[4])
  \Edge[label = $4.3$](c[4])(a[4])
  \Edge[label = $5.3$](a[3])(a[5])
  \Edge[label = $6.1$](b[2])(c[1])

  \end{tikzpicture} 
\\
  \hline
\end{tabular}

Based on the corollary of theorem 2, we record the following results:
\begin{enumerate}
\item There exists \textit{flow dependence} in $G$ if for node $v \in V $ in $G$. If there is an edge $L((u,v))=s.i$ and an edge $L((v,u^`))=t.j$, where $j > i$. $ m \leq i,j \leq n $ where $m,n$ represent the loop bounds. $s,t$ are the indices of loop instructions. 
\item There exists \textit{Anti dependence} in $G$ if for node $v \in V $ in $G$. If there is an edge $L((u,v))=s.i$ and an edge $L((v,u^`))=t.j$, where $j < i$. $ m \leq i,j \leq n $, where $m,n$ represent the loop bounds. $s,t$ are the indices of loop instructions. 
\item There exists \textit{output dependence} in $G$ if for node $v \in V $ in $G$. If there are two  edges $L((u,v))=s.i$ and $L((u^`,v))=t.j$ where $t \neq s$ and $j \neq i$. 
\item There exists \textit{Input dependence} in $G$ if for node $v \in V $ in $G$. If there are two  edges $L((v,u))=s.i$ and $L((v,u^`))=t.j$ where $t \neq s$ and $j \neq i$.
\end{enumerate}
\textbf{Data Dependence in Two Dimensional Arrays} \\
The concept discussed for single loops is extended to nested loops. Example 3 illustrates the identification of data dependencies in a program with two dimensional arrays.

\begin{tabular}{ l c r  }
 \hline
 \multicolumn{2}{c}{Example 3} \\
 \hline
 (a) & (b) \\
 \hline
 \begin{lstlisting} [mathescape]
void add()  
{
for(i=1;i<3;i++)
   for(j=1;j<3;j++)
       $L_1:$ a[i][j]=c[i][j-1];
       $L_2:$ c[i][j]=a[i-1][j];
}
\end{lstlisting}   &\begin{tikzpicture}[scale=1.8]
\GraphInit[vstyle=Normal]
%\tikzset{ LabelStyle/.style = { rectangle, rounded corners, draw}}
\SetUpEdge[style={->}]
  \SetGraphUnit{3}
  \Vertex {a[1][1]}
  \Vertex[x=-0.3,y=-0.9]{a[1][2]}
    \Vertex[x=-0.6,y=-1.9]{a[2][1]}
        \Vertex[x=3.8,y=-0.3]{a[2][2]}

  \Vertex[x=1.2,y=0.0]{c[1][1]}
  \Vertex[x=2.3,y=-1.3]{a[0][1]}
  \Vertex[x=1.5,y=-1.5]{c[2][1]}
        \Vertex[x=1.4,y=-2.7]{c[2][2]}
        \Vertex[x=2.2,y=0.2]{c[0][0]}
        \Vertex[x=3.2,y=-0.9]{c[1][2]}
        \Vertex[x=2.5,y=-2.1]{c[0][1]}
                \Vertex[x=2.6,y=-3.1]{a[0][2]}
                \Vertex[x=2.0,y=-3.7]{c[1][0]}

          %      \Vertex[x=3.7,y=-0.2]{b[5]}

% \tikzset{VertexStyle/.style = {shape=circle,fill=black}}
 %\Vertex [x=0,y=1.0]{s}
 %\tikzset{VertexStyle/.style = {shape=rectangle,fill=black}}
 %\Vertex [x=0.8,y=-1.0]{s1}
 
   \Edge[label = $L_2.3$](a[1][1])(c[2][1])

    \tikzset{EdgeStyle/.style = { ->,bend left}}
  \Edge[label = $L_1.2$](c[0][1])(a[1][2])
  \Edge[label = $L_1.3$](c[1][0])(a[2][1])
    \Edge[label = $L_2.1$](a[0][1])(c[1][1])

   \tikzset{EdgeStyle/.style = { ->,bend right}}
  
    \Edge[label = $L_1.1$](c[0][0])(a[1][1])
  \Edge[label = $L_2.4$](a[1][2])(c[2][2])
      \Edge[label = $L_2.2$](a[0][2])(c[1][2])
    \Edge[label = $L_1.4$](c[1][1])(a[2][2])

  \end{tikzpicture} 
\\
  \hline
\end{tabular}

\subsection{Data dependence in scalars}
\par Scalars are variables used to store a single data. Data dependence information of scalar variables will help to group the program instructions that can be executed in parallel. Identifying data dependencies in scalar variables present in loops will help to determine whether a loop can be parallelized or not. Here, we discuss the process by which our DDI model will identify the scalar dependencies present in loops.

\begin{tabular}{ l c r  }
 \hline
 \multicolumn{2}{c}{Example 4} \\
 \hline
 (a) & (b) \\
 \hline
 \begin{lstlisting} [mathescape]
void add()  
{
$1:$ c=0;
for($\overbrace{i=1}^{2}$;$\overbrace{i<3}^{3}$;$\overbrace{i++}^{4}$)
$5:$ s=c+i;;
$6:$ c=s;
}
\end{lstlisting}   &\begin{tikzpicture}[scale=1.8]
\GraphInit[vstyle=Normal]
%\tikzset{ LabelStyle/.style = { rectangle, rounded corners, draw}}
\SetUpEdge[style={->}]
  \SetGraphUnit{3}
  \Vertex {c}
  \Vertex[x=-1.2,y=-0.9]{s}
    \Vertex[x=0,y=-1.5]{i}
          \Vertex[x=0.5,y=-0.9]{PR}
  \Edge[label = $1$](PR)(c)

  \Edge[label = $5.1$](i)(s)
  \Edge[label = $6.1$](s)(c)
      \tikzset{EdgeStyle/.style = { ->,bend left}}
    \Edge[label = $5.2$](c)(s)
  \Edge[label = $5.2$](i)(s)
        \tikzset{EdgeStyle/.style = { ->,bend left=25}}

  \Edge[label = $6.2$](s)(c)

   \tikzset{EdgeStyle/.style = { ->,bend right=60}}
    \Edge[label = $5.1$](c)(s)

    \end{tikzpicture} 
\\
  \hline
\end{tabular}

There exists flow dependence in scalar variable present in loop if the variable is read as well as  written in every instance of loop iteration. For a node $v$, if there exists an incoming as well as outgoing edge with labels of same loop instance then the flow dependence exists in a scalar variable in a loop. In Example 4, variables $c$ and $s$ are accessed for `Read' and `Write' in the same loop instance.
\subsection{Observations on existing data dependence tests}
 Many data dependence tests \cite{kong},\cite{uptal},\cite{omega},\cite{psarris} has been proposed for arrays. Almost all the tests are based on subscript analysis. If the array subscripts are linear, they are converted to the form linear equations and inequalities. The problem is about solving these linear equations and inequalities to check whether an integer solution exists. Existence of integer solution proves data dependence. The problem is more complicated in case of multidimensional arrays as it ends up in solving system of linear equations and inequalities.
\par There is always a trade-off between accuracy and complexity in the existing data dependence tests. Table 1 shows the experimental results from \cite{test1} on perfect benchmarks \cite{perfect}. The table displays the total number of data dependence problems, number of problems proven to be independent i.e., no data dependence exist by each test, number of problems where the test could not conclude whether data dependence exists or not, average time taken by each test per dependence problem to conclude whether data dependence exists or not. Omega test was able to prove that 35\% of the problems are not data dependent, where as Banerjee and I Test were able to do it for 30\% alone. On the other hand Omega test took 177ms on average per dependence problem where as Banerjee and I Test took only 8ms to perform dependence testing.
\begin{table}
\centering
\begin{tabular}{ |p{2cm}|p{3.5cm}|p{2cm}|p{2cm}|p{2cm}|  } 
 \hline
  \multicolumn{5}{|c|}{Perfect Benchmarks} \\
 \hline
 Test& Data Dependence Problems	 & Independent & May or may not be dependent& Time(msec)  \\
 \hline
   Banerjee   & 59936  & 17946(30\%) & 41990(70\%) & 8 \\
      \hline
      I-Test & 59936  & 17970(30\%) & 37827(63\%) & 8 \\
 \hline
      Omega Test & 59936  & 21232(35\%) & 32239(54\%) & 177  \\
 
 \hline
\end{tabular}
\caption{Comparison of Data Dependence Tests w.r.t accuracy and time}
\end{table}
\par We have theoretically proved(Theorem 2) that our DDI model can identify all kinds of data dependencies in a program in polynomial time with out any error.
\section{Basic Transformation Techniques}
 Basic transformations like constant propagation, dead code elimination and induction variable detection are performed on a program to optimize it. In this section, we will discuss how these transformations can be performed using our DDI model.  
\subsection{Dead code Elimination}
Dead code refers to the variables whose data is never used in the program. Removing such code in the program reduces the program's size. All compilers perform dead code elimination as a part of compiler optimization. Here we will discuss how our model identifies and eliminates dead code in a given program.
The following scenarios depict the presence of dead code in a given program:
\begin{itemize}
\item A variable has been initialized but it is never used in the program's execution.
\item Data computed in the program, but it is never been used to get the final output.
\end{itemize}
The following characteristics in the graph shows the existence of dead code:
\begin{enumerate}
    \item For a node $u$, there are no edges $(v,u)$ and $(u,u^`)$ where $ v,u^` \in N.G_P$.
    \item For a node $u$, there are edges $(v,u)$ but no edges $(u,u^`)$ where $ v,u^` \in N.G_P$.
    \item For a node $u$, there are edges where $L((v,u))=i$,$L((v,u))=j$ and $L((u,u^`))=k$ but no edges where $L((u,u^`))>i\; and <j$, $ v,u^` \in N.G_P$.
    \item For a node $u$, there are edges where $L((v,u))=i$,$L((v,u))=j$ and $L((u,u^`))=k$ but no edges where $L((u,u^`))>j$, $v,u^` \in N.G_P$.
\end{enumerate}
Algorithm illustrates the process of dead code elimination.
\begin{algorithm}
\caption{Dead code Elimination}
\begin{algorithmic}[1]
\Procedure{Dead code Elimination}{}\newline   \label{euclid}
\textbf{Input:} Graph $G_p(N,E,L)$
 \For{every $v \in N.G$}
   \If{$(L((u,v)) == k) \&\& (L((v,u^1)) == NULL) $}
          \State{remove $k$}
    \EndIf
    \If{$(L((u,v)) == k) \&\& (L((v,u^1)) < k) $}
          \State{remove $k$}
            \EndIf
    \If{$(L((u,v)) == k) \&\& L((u,v)) == m) \&\& (L((v,u^1)) > k \&\& <m == NULL) $}
          \State{remove $k$}
            \EndIf

   \If{$(L((u,v)) == NULL) \&\& (L((v,u^1)) == NULL) $}  
          \State{remove $v$}
       
    \EndIf
 \EndFor
\EndProcedure
\end{algorithmic}
\end{algorithm}

Lines 9-10 of algorithm 3 identifies a node $v$ that doesn't have any incoming or outgoing edges which elucidate that a variable is initialized in the program but no data is assigned and never used in the program. Such variables are considered as unused and are removed in the program. Lines 3-4 identifies a node $v$ that have only incoming edges but no outgoing edges which explains that data is assigned to a variable but that variable is never used in the program. In example 5, in the program one could observe that the variable $c$ is never used in any instruction of the program. In $G_P$, an edge $L((PR,c))=1$ tells that a constant value is written to $c$ but $c$ doesn't have any outgoing edges which elucidate that variable $c$ is never read in the program. In instruction $2:[\{b,PR\},\{a\} ]$, $\forall I:[R,W], a\notin R $ i.e. $a$ is computed in instruction 2 but in no other instruction $a$ is read. In graph, $L((u,a)) \neq NULL $ and $L((a,u^1))=NULL$ which tells node $a$, is having only incoming edges but no outgoing edges means variable $a$ is not read any where in the program. Therefore, value computed in instruction 2 and variable $c$ is never used for output computation in the program.

\begin{tabular} { l c r }
 \hline
 \multicolumn{2}{c}{Example 5} \\
 \hline            
 (a) & (b) \\
 \hline
\begin{lstlisting}[mathescape=true]
void func1(int a,int b)
{
$1:$ int a,b=3,c=5,d;
$2:$ a=b+5;
$3:$ d=b*10;
$4:$ print d;
}
\end{lstlisting}   & 

\begin{tikzpicture}[scale=2]
\GraphInit[vstyle=Normal]
%\tikzset{ LabelStyle/.style = { rectangle, rounded corners, draw}}
\SetUpEdge[style={->}]
  \SetGraphUnit{3}
  \Vertex{b}
  \Vertex[x=-1.0,y=-0.3]{a}

  \Vertex[x=0.4,y=-2.0]{c}
    \Vertex[x=0,y=-1.2]{d}
  \Vertex[x=-0.6,y=-2.0]{HU}
  \Vertex[x=0.8,y=-0.6]{PR}
 % \WE(B){A}
 % \EA(B){C}
  \Edge[label = $1$](PR)(c)
    \Edge[label = $1$](PR)(b)

  \Edge[label = $2$](b)(a)
   \Edge[label = $3$](b)(d)
          \Edge[label = $3$](PR)(d)
   \Edge[label = $4$](d)(HU)

   \tikzset{EdgeStyle/.style = { ->,bend left}}
   \Edge[label = $2$](PR)(a)

  \end{tikzpicture}   
\\  \hline
\end{tabular}

\begin{tabular} { l c r }
 \hline
 \multicolumn{2}{c}{Example 6} \\
 \hline            
 (a) & (b) \\
 \hline
\begin{lstlisting}[mathescape=true]
void func1(int a,int b)
{
$1:$ int a,b=3,c=5,d;
$2:$ a=b+5;
$3:$ d=b*c;
$4:$ a=b+c;
$5:$ print a,d;
}
\end{lstlisting}   & 

\begin{tikzpicture}[scale=2]
\GraphInit[vstyle=Normal]
%\tikzset{ LabelStyle/.style = { rectangle, rounded corners, draw}}
\SetUpEdge[style={->}]
  \SetGraphUnit{3}
  \Vertex{b}
  \Vertex[x=-1.0,y=-0.3]{a}
  \Vertex[x=0.4,y=-2.0]{c}
    \Vertex[x=0,y=-1.2]{d}
  \Vertex[x=-0.6,y=-2.0]{HU}
  \Vertex[x=0.8,y=-0.6]{PR}
 % \WE(B){A}
 % \EA(B){C}
  \Edge[label = $1$](PR)(c)
    \Edge[label = $1$](PR)(b)

  \Edge[label = $2$](b)(a)
   \Edge[label = $3$](b)(d)
          \Edge[label = $3$](c)(d)
   \Edge[label = $5$](d)(HU)
      \Edge[label = $5$](a)(HU)
  \Edge[label = $4$](c)(a)

   \tikzset{EdgeStyle/.style = { ->,bend left}}
   \Edge[label = $2$](PR)(a)
   \tikzset{EdgeStyle/.style = { ->,bend right}}
  \Edge[label = $4$](b)(a)

  \end{tikzpicture}   
\\  \hline
\end{tabular}

Lines 7-8 identifies the instruction that is unused. Node $v$ have incoming edges $k$ and $m$ and no outgoing edge between $k$ and $m$ which says that data computed in instruction $k$ is unused. Instruction $k$ can be removed. In example 6, variable $a$ is computed in instruction $2:[\{b,PR\} ,\{a\} ]$ and $4:[\{b,c\} ,\{a\} ]$ but for $I:[R,W],a \notin R\; for\; I>2\; and \;I<4$. Similar structure can be observed in the graph that variable $a$ have incoming edges with labels 2 and 4 but no outgoing edges in between 2 and 4. Instruction 2 can be removed.

%\end{itemize}
%Algorithm illustrates the process of dead code elimination.
%\begin{algorithm}
%\caption{Dead code Elimination}
%\begin{algorithmic}[1]
%\Procedure{Dead code Elimination}{}\newline   \label{euclid}
%\textbf{Input:} Graph $G_p(N,E,L)$
% \For{every $v \in N.G$}
 %  \If{$(L(u,v) == k) \&\& (L(v,u^1) == NULL) $}
  %        \State{remove $k$}
   % \EndIf
    %\If{$(L(u,v) == k) \&\& (L(v,u^1) < k $}
     %     \State{remove $k$}
      %      \EndIf

  %\If{$(L(u,v) == NULL) \&\& (L(v,u^1) == NULL) $}
   
   %       \State{remove $v$}

   % \EndIf
 
 %\EndFor

%\EndProcedure
%\end{algorithmic}
%\end{algorithm}
The running time of this algorithm is $O(|N|^2)$ as the entire row and column of each variable $v$ has to be scanned.  
\subsection{Constant Propagation}
In a program, if a variable is assigned with a constant value, references to this variable in the program can be replaced directly with the constant value. This technique is called as constant propagation. Here, we will discuss how this technique is designed using our model.

\begin{algorithm}
\caption{Constant Propagation}
\begin{algorithmic}[1]
\Procedure{Constant Propagation}{}\newline   \label{euclid}
\textbf{Input:} Graph $G_p(N,E,L)$
 \For{every $v \in N.G$}
   \If{$(L((PR,v)) == k) \&\& (L((u^1,v)) == NULL) $}
      \For{every $L((v,u))=j$}
           \State{add edge $L((PR,u))=k$}
                      \State{delete edge $L((PR,v))=k \; and \; L((v,u))=j$}
 \EndFor
    \EndIf   
   \If{$(L((PR,v)) == k) \&\& (L((u^1,v)) \neq NULL) $}
      \For{every $L((v,u))=j \; where \; j>k \; and \; <m$}
           \State{add edge $L((PR,u))=k$}
                      \State{delete edge $L((PR,v))=k \; and \; L((v,u))=j$}
 \EndFor
    \EndIf    
 \EndFor
\EndProcedure
\end{algorithmic}
\end{algorithm}
Constant Propagation requires 1) identification of variable assigned with constant value. 2) Replacing the variable directly with the constant value where it has been referenced in the followed up instructions of the program. 
\par In a given instruction $I:[R,W]$, a constant value is assigned to a variable if $r=PR$ i.e. R is having only one variable PR, which represents constant value. In the graph, for a node $v$, if there is an incoming edge from node PR with label $k$ and no other incoming edges to node $v$ have the same label $k$, it says that in instruction $k$ a constant value is assigned to a variable. 
\par Lines 3-6 of algorithm 4, identifies a node $u$ where $L(PR,u)=k$ and no other incoming edges exists to $u$. For every outgoing edge of $u$ i.e. $L((u,v))=j$, where $j>k$, add edge $L((PR,v))=j$ and remove edges $L((PR,u))=k$ and $L((u,v))=j$ in $G_p$. In example 7, in instruction 1, $b$ is assigned with a constant value 3. Variable $b$ is read in instruction 2 which can be directly replaced with 3. 
 \par Lines 7-10 of algorithm, identifies a node $u$ where $L((PR,u))=k$ and other incoming edges $L((u^1,u))=m$ exists to $u$. For every outgoing edge of $u$ i.e. $L((u,v))=j$, where $j>k$ and $j<m$, add edge $L((PR,v))=j$ and remove edges $L((PR,u))=k$ and $L((u,v))=j$ in $G_p$. 
 
 \begin{tabular} { l c r }
 \hline 
 \multicolumn{2}{c}{Example 7} \\
 \hline            
 (a) & (b) \\
 \hline
\begin{lstlisting}[mathescape=true]
void func1(int a,int b)
{
$1:$ int a,b=3,c;
$2:$ a=b+5;
$3:$ c=a;
$4:$ print c;
}
\end{lstlisting}   & 
\begin{tikzpicture}[scale=2]
\GraphInit[vstyle=Normal]  
%\tikzset{ LabelStyle/.style = { rectangle, rounded corners, draw}}
\SetUpEdge[style={->}]
  \SetGraphUnit{3}
  \Vertex{b}
  \Vertex[x=-1.0,y=-0.3]{a}
  \Vertex[x=0.4,y=-2.0]{c}
  \Vertex[x=-0.6,y=-2.0]{HU}
  \Vertex[x=0.8,y=-0.6]{PR}
 % \WE(B){A}
 % \EA(B){C}
    \Edge[label = $1$](PR)(b)
  \Edge[label = $2$](b)(a)
   \Edge[label = $3$](c)(a)
      \Edge[label = $4$](c)(HU)
   \tikzset{EdgeStyle/.style = { ->,bend left}}
   \Edge[label = $2$](PR)(a)
  \end{tikzpicture}   
\\  \hline
\end{tabular}

\subsection{Induction variable Detection}
An induction variable is a variable in loop whose value either increments or decrements constantly in every iteration. There are two types of induction variables: basic and derived. Basic induction variables are of the form $x=x+c$ where $c$ is constant value and $x$ is basic induction variable. Derived induction variable is a variable defined in the loop which is a linear function of basic induction variable.
Here, we will discuss how our model detects basic and derived induction variables.
 
\begin{tabular}{ l c r  }
\hline
\multicolumn{2}{c}{Example 8} \\
\hline
(a) & (b) \\
\hline
\begin{lstlisting} [mathescape]
void add()  
{
$1:$ s=0;
for($\overbrace{i=1}^{2}$;$\overbrace{i<2}^{3}$;$\overbrace{i++}^{4}$)
$5:$ s=s+i;;
$6:$ c=s*5;
}
\end{lstlisting}   &\begin{tikzpicture}[scale=1.8]
\GraphInit[vstyle=Normal]
\tikzset{ LabelStyle/.style = { rectangle, rounded corners, draw}}
\SetUpEdge[style={->}]
 \SetGraphUnit{3}
\Vertex {c}
\Vertex[x=-1.2,y=-0.9]{s}
\Vertex[x=-0.3,y=-1.9]{i}
 \Vertex[x=0.5,y=-0.9]{PR}
 \Vertex[x=0.5,y=-1.9]{HU}

\Edge[label = $6.1$](PR)(c)
\Edge[label = $1$](PR)(s)

\Edge[label = $2$](PR)(i)
\Edge[label = $5.1$](i)(s)
\Edge[label = $3$](PR)(HU)

\Edge[label = $6.1$](s)(c)
  \tikzset{EdgeStyle/.style = { ->,loop}}
 \Edge[label = $5.1$](s)(s)
   \tikzset{EdgeStyle/.style = { ->,loop below}}

  \Edge[label = $4$](i)(i)

   \tikzset{EdgeStyle/.style = { ->,bend right}}
\Edge[label = $4$](PR)(i)

\end{tikzpicture} 
\end{tabular}

If a node $v$ in graph $G_p$, consists of self loop then $v$ is a basic induction variable.\\
For a node $u$, if there is an incoming edge from $v$ i.e. the basic induction variable then $u$ is a derived induction variable. 
\section{Data Dependence in pointers}
A pointer is a variable that stores the address of another variable. Identifying data dependencies in a program that uses pointers is solved using alias analysis\cite{alias} and shape analysis \cite{shape}. Here, we will discuss how to represent and identify data dependencies in pointers using DDI model.

\begin{tabular} { l c r }
 \hline
 \multicolumn{2}{c}{Example 9} \\
 \hline            
 (a) & (b) \\
 \hline
\begin{lstlisting}[mathescape=true]
void poin1(int a,int b)
{
$1:$ int a=3,*p,c;
$2:$ p=&a;
$3:$ c=*p+a;
$4:$ print p,*p,c;
}
\end{lstlisting}   & 

\begin{tikzpicture}[scale=2]
\GraphInit[vstyle=Normal]
%\tikzset{ LabelStyle/.style = { rectangle, rounded corners, draw}}
\SetUpEdge[style={->}]
  \SetGraphUnit{3}
  \Vertex{a}
  \Vertex[x=-1.0,y=-0.3]{PR}
  \Vertex[x=0,y=-1.0]{c}
      \Vertex[x=1.2,y=-0.2]{p}
      \Vertex[x=1.2,y=-2.0]{HU}

  \Edge[label = $1$](PR)(a)

  \Edge[label = $3$](a)(c)
     \tikzset{EdgeStyle/.style = { ->,bend left}}
   \Edge[label = $3$](a)(c)
          \Edge[label = $4$](c)(HU)
   \Edge[label = $4$](p)(HU)
      \tikzset{EdgeStyle/.style = { ->,bend right=75}}

      \Edge[label = $4$](a)(HU)
      \tikzset{EdgeStyle/.style = {dashed, ->}}

    \Edge[label = $2$](a)(p)
  \end{tikzpicture}   
\\  \hline
\end{tabular}

Pointers are variables that stores the address of another variable, they are represented as usually as how variables are represented in our model. The edge for the assignment statement, which assigns the address of a variable to pointer variable is represented using dashed line and added to graph. In example 8, pointer assignment statement $p=\&a$ is  represented as dashed directed edge in the graph with label 2. \\
Data dependence exists if there is a directed path of length $ \geq 2$ in graph $G_p$. The same theorem applies for pointers with the constraint  that the dashed edges should not be included in the path.

\section{Conclusion}
\par Thus, in summary, we have proposed a graph based Data Dependence Identifier (DDI) which will identify all types of data dependencies that include scalars, arrays and pointers. Further, compiler optimization techniques like dead code elimination, constant propagation, and induction variable detection can be performed using DDI.\\
\par Some of the salient features of our work are:
\begin{enumerate}
\item Before representing the program as graph, we parameterized  the program with components $(I,V \cup \{ HU,PR \} ,MAI)$. 
\item To represent the program as graph, we considered variables $V$ as set of nodes in the graph and the edges are drawn based on the way variables are accessed from memory.
\item Our Data Dependence Identifier(DDI) is a unique model that can identify all types of data dependencies in polynomial time. 
\item Compiler optimizations like constant propagation, induction variable detection, dead code elimination can be performed using DDI.
\end{enumerate}
\par By the application of graph theoretical concepts, our DDI model can be further investigated to break the given program into parallelized modules. To reduce the time complexity of DDI model, further investigation can be initiated to solve the data dependence identification problem using graph based machine learning algorithms. 

\bibliography{mybibfile}

\end{document}